\documentclass[a4paper,UKenglish,cleveref, autoref, thm-restate]{lipics-v2021}

\usepackage{tikz}
\usetikzlibrary{automata,positioning,arrows}

\title{Random Wheeler Automata} 

\author{Ruben {Becker}}{Ca' Foscari University of Venice, Italy}{rubensimon.becker@unive.it}{https://orcid.org/0000-0002-3495-3753}{}
\author{Davide {Cenzato}}{Ca' Foscari University of Venice, Italy}{davide.cenzato@unive.it}{https://orcid.org/0000-0002-0098-3620}{}
\author{Sung-Hwan {Kim}}{Ca' Foscari University of Venice, Italy}{sunghwan.kim@unive.it}{https://orcid.org/0000-0002-1117-5020}{}
\author{Bojana {Kodric}}{Ca' Foscari University of Venice, Italy}{bojana.kodric@unive.it}{https://orcid.org/0000-0001-7242-0096}{}
\author{Riccardo {Maso}}{Ca' Foscari University of Venice, Italy}{riccardomaso27@gmail.com}{}{}
\author{Nicola {Prezza}}{Ca' Foscari University of Venice, Italy}{nicola.prezza@unive.it}{https://orcid.org/0000-0003-3553-4953}{}

\authorrunning{R. Becker et al.} 

\Copyright{Ruben Becker, Davide Cenzato, Sung-Hwan Kim, Bojana Kodric, Riccardo Maso and Nicola Prezza} 



\category{} 

\relatedversion{} 

\newtheorem{problem}[theorem]{Problem}

\funding{\textit{Ruben Becker, Davide Cenzato, Sung-Hwan Kim, Bojana Kodric, Nicola Prezza}: Funded by the European Union (ERC, REGINDEX, 101039208). Views and opinions expressed are however those of the author(s) only and do not necessarily reflect those of the European Union or the European Research Council. Neither the European Union nor the granting authority can be held responsible for them.}

\DeclareMathOperator{\rank}{rank}

\DeclareMathOperator{\mask}{mask}

\let\epsilon\varepsilon
\let\eps\varepsilon

\usepackage{dsfont}

\usepackage{xspace}

\usepackage[linesnumbered,ruled,vlined,boxed]{algorithm2e}
\newlength{\commentWidth}
\let\oldnl\nl
\newcommand{\nonl}{\renewcommand{\nl}{\let\nl\oldnl}}
\DontPrintSemicolon
\SetKwInOut{Input}{Input}\SetKwInOut{Output}{Output}

\usepackage{tikz}
\usetikzlibrary{decorations.pathreplacing}
\usetikzlibrary{plotmarks}
\usetikzlibrary{positioning,automata,arrows}
\usetikzlibrary{shapes.geometric}
\usetikzlibrary{decorations.markings}
\usetikzlibrary{positioning, shapes, arrows}

\PassOptionsToPackage{usenames,dvipsnames,svgnames}{xcolor}
\definecolor{orange}{RGB}{235,90,0}
\definecolor{darkorange}{RGB}{175,30,0}
\definecolor{turkis}{RGB}{131,182,182}
\definecolor{darkturkis}{RGB}{31,82,82}
\definecolor{green}{RGB}{102,180,0}
\definecolor{darkgreen}{RGB}{51,90,0}
\definecolor{myblue}{RGB}{0,0,213}
\definecolor{mydarkblue}{RGB}{0,0,100}
\definecolor{mybrightblue}{HTML}{74B0E4}
\definecolor{mybrighterblue}{HTML}{B3EAFA}
\definecolor{lila}{RGB}{102,0,102}
\definecolor{darkred}{RGB}{139,0,0}
\definecolor{darkyellow}{RGB}{188,135,2}
\definecolor{brightgray}{RGB}{200,200,200}
\definecolor{darkgray}{RGB}{50,50,50}
\definecolor{amaranth}{rgb}{0.9, 0.17, 0.31}
\definecolor{alizarin}{rgb}{0.82, 0.1, 0.26}
\definecolor{amber}{rgb}{1.0, 0.75, 0.0}
\definecolor{green(ryb)}{rgb}{0.4, 0.69, 0.2}
\definecolor{hanblue}{rgb}{0.27, 0.42, 0.81}
\definecolor{grannysmithapple}{rgb}{0.66, 0.89, 0.63}

\newcommand{\ArxivOrCr}[2]{#1}

\ArxivOrCr{\EventEditors{John Q. Open and Joan R. Access}
\EventNoEds{2}
\EventLongTitle{42nd Conference on Very Important Topics (CVIT 2016)}
\EventShortTitle{CVIT 2016}
\EventAcronym{CVIT}
\EventYear{2016}
\EventDate{December 24--27, 2016}
\EventLocation{Little Whinging, United Kingdom}
\EventLogo{}
\SeriesVolume{42}
\ArticleNo{23}
}{
\EventEditors{Shunsuke Inenaga and Simon J. Puglisi}
\EventNoEds{2}
\EventLongTitle{35th Annual Symposium on Combinatorial Pattern Matching (CPM 2024)}
\EventShortTitle{CPM 2024}
\EventAcronym{CPM}
\EventYear{2024}
\EventDate{June 25--27, 2024}
\EventLocation{Fukuoka, Japan}
\EventLogo{}
\SeriesVolume{296}
\ArticleNo{22}
}

\begin{document}

\maketitle

\begin{abstract}
Wheeler automata were introduced in 2017 as a tool to generalize existing indexing and compression techniques based on the Burrows-Wheeler transform. 
Intuitively, an automaton is said to be Wheeler if there exists a total order on its states reflecting the natural co-lexicographic order of the strings labeling the automaton's paths; this property makes it possible to represent the automaton's topology in a constant number of bits per transition, as well as efficiently solving pattern matching queries on its accepted regular language. 
After their introduction, Wheeler automata have been the subject of a prolific line of research, both from the algorithmic and language-theoretic points of view. A recurring issue faced in these studies is the lack of large datasets of Wheeler automata on which the developed algorithms and theories could be tested. 
One possible way to overcome this issue is to generate random Wheeler automata. 
Motivated by this observation of practical nature, in this paper we initiate the theoretical study of random Wheeler automata, focusing our attention on the deterministic case (Wheeler DFAs --- WDFAs). We start by naturally extending the Erdős-R\'enyi random graph model to WDFAs, and proceed by providing an algorithm generating uniform WDFAs according to this model. 
Our algorithm generates a uniform WDFA with
$n$ states, $m$ transitions, and alphabet's cardinality $\sigma$
in $O(m)$ expected time ($O(m\log m)$ time w.h.p.) and constant working space for all alphabets of size $\sigma \le m/\ln m$. 
The output WDFA is streamed directly to the output.
As a by-product, we also give formulas for the number of distinct WDFAs and obtain that $ n\sigma + (n - \sigma) \log \sigma$ bits are necessary and sufficient to encode a WDFA with $n$ states and alphabet of size $\sigma$, up to an additive $\Theta(n)$ term.
We present an implementation of our algorithm and show that it is extremely fast in practice, with a throughput of over 8 million transitions per second.
\end{abstract}

\clearpage
\setcounter{page}{1}

\section{Introduction}

Wheeler automata were introduced by Gagie et al.\ in \cite{gagie:tcs17:wheeler} in an attempt to unify existing indexing and compression techniques based on the Burrows-Wheeler transform \cite{burrows1994block}. An automaton is said to be Wheeler if there exists a total order of its states such that (i) states reached by transitions bearing different labels are sorted according to the underlying total alphabet's order, and (ii) states reached by transitions bearing the same label are sorted according to their predecessors (i.e. the order propagates forward, following pairs of equally-labeled transitions). 
Equivalently, these axioms imply that states are sorted according to the co-lexicographic order of the strings labeling the automaton's paths. 
Since their introduction, Wheeler automata have been the subject of a prolific line of research, both from the algorithmic \cite{ConteMatchingStatistics2023,alanko2019tunneling,gagie2022representing,gibney2022complexity,egidi2022space,chao2022wgt,goga2022prefix} and language-theoretic \cite{ALANKO2021104820,AlankoRegular2020,DAgostinoOrdering2023} points of view. 
The reason for the success of Wheeler automata lies in the fact that their total state order enables \emph{simultaneously} to index the automaton for pattern matching queries and to represent the automaton's topology using just $O(1)$ bits per transition (as opposed to the general case, requiring a logarithmic number of bits per transition). 

A recurring issue faced in research works on Wheeler automata is the lack of datasets of (large) Wheeler automata on which the developed algorithms and theories could be tested. As customary in these cases, a viable solution to this issue is to randomly generate the desired combinatorial structure, following a suitable distribution. The most natural distribution, the uniform one, represents a good choice in several contexts and can be used as a starting point to shed light on the combinatorial objects under consideration; the case of random graphs generated using the Erdős-R\'enyi random graph model \cite{nicaud14:randomDFA} is an illuminating example. 
In the case of Wheeler automata, we are aware of only one work addressing their random generation: the WGT suite  \cite{chao2022wgt}. This random generator, however, does not guarantee a uniform distribution over the set of all Wheeler automata. 

\subsection{Our contributions}\label{sec:contributions}

Motivated by the lack of formal results in this area, in this paper we initiate the theoretical study of random Wheeler automata, 
focusing our attention on the algorithmic generation of uniform deterministic Wheeler DFAs (WDFAs). 
We start by extending the Erdős-R\'enyi random graph model 
to WDFAs: our uniform distribution is defined over the set $\mathcal D_{n, m, \sigma}$ of all Wheeler DFAs over the \emph{effective}  alphabet  (i.e. all labels appear on some edge) $[\sigma] = \{1,\dots, \sigma\}$, 
with $n$ states $[n]$, $m$ transitions, and Wheeler order $1 < 2 < \ldots < n$. 
We observe that, since any WDFA can be encoded using $O(n\sigma)$ bits \cite{gagie:tcs17:wheeler}, the cardinality of $\mathcal D_{n, m, \sigma}$ is at most $2^{O(n\sigma)}$. On the other hand, the number of DFAs with $n$ states over alphabet of size $\sigma$ is $2^{\Theta(n\sigma \log n)}$ \cite{nicaud14:randomDFA}. As a result, a simple rejection sampling strategy that uniformly generates DFAs until finding a WDFA (checking the Wheeler property takes linear time on DFAs \cite{AlankoRegular2020}) would take expected exponential time to terminate. 
To improve over this naive solution, we start by  defining a new combinatorial characterization of WDFAs: in Section \ref{sec:uniform WDFAs}, we establish a bijection that associates every element of $\mathcal D_{n, m, \sigma}$ to a pair formed by a binary matrix and a binary vector. 
This allows us to design an algorithm to uniformly sample WDFAs, based on the above-mentioned representation. 
Remarkably, our sampler uses \emph{constant} working space and streams the sampled WDFA directly to output:

\begin{theorem}
\label{th:main theorem}
There is an algorithm to generate a uniform WDFA from $\mathcal D_{n, m, \sigma}$ in $O(m)$ expected time ($O(m\log m)$ time with high probability) using $O(1)$ words of working space, for all alphabets of size $\sigma \le m/\ln m$. The output WDFA is directly streamed to the output as a set of labeled edges.
\end{theorem}

As a by-product of our combinatorial characterization of WDFAs, in Theorem \ref{th: number of WDFAs} 
we give an exact formula for the number $|\mathcal D_{n, m, \sigma}|$ of distinct WDFAs with $n$ nodes and $m$ edges labeled from alphabet $[\sigma]$ and in Theorem \ref{thm:bound n sigma} we give a tight asymptotic formula for the number $|\mathcal D_{n, \sigma}|$ of distinct WDFAs with $n$ nodes and any number of edges labeled from $[\sigma]$, obtaining that $ n\sigma + (n - \sigma) \log \sigma$ bits are necessary and sufficient to encode WDFAs from such a family up to an additive $\Theta(n)$ term.

We conclude by presenting an implementation of our algorithm, publicly available at \url{https://github.com/regindex/Wheeler-DFA-generation}, and showing that it is very fast in practice while using a negligible (constant) amount of working space.

\section{Preliminaries and Problem Statement}
With $\ln x$ and $\log x$, we indicate the natural logarithm and the logarithm in base 2 of $x$, respectively. 
For an integer $k\in\mathbb{N}^{+}$, we let $[k]$ denote the set of all integers from 1 to $k$. For a bit-vector $v\in \{0,1\}^k$, we denote with $\|v\| = \sum_{i\in [k]} v_i$ the $L_1$-norm of $v$, i.e., the number of set bits in $v$. For an integer $\ell\le k$, we denote with $v[1:\ell]$ the bit-vector $(v_1, \ldots, v_\ell)$ consisting only of the first $\ell$ bits of $v$. 
For a bit-matrix $A\in \{0, 1\}^{\ell \times k}$ and a column index $j\in [k]$, we denote the $j$'th column of $A$ by $A_j$
and the element at row $i$ and column $j$ as $A_{i,j}$. We let $\|A\| = \sum_{i\in [k], j\in [\ell]} A_{i,j}$ be the $L_{1,1}$-norm of $A$, which again counts the number set bits in $A$.
For a bit-vector $v\in \{0,1\}^k$, we use the notation $\rank(v, i)$ to denote the number of occurrences of $1$ in $v[1:i]$. For completeness, we let $\rank(v, 0) = 0$. We generalize this function also to matrices as follows. For a bit-matrix $A\in \{0, 1\}^{\ell \times k}$, we let $\rank(A, (i, j))=\sum_{r\in [j-1]}\rank(A_r, \ell) + \rank(A_{j}, i)$. We sometimes write bit-vectors from $\{0,1\}^k$ in string form, i.e., as a sequence of $k$ bits.


In this paper we are concerned with deterministic finite automata.

\begin{definition}[Determinisitic Finite Automaton (DFA)]\label{def:DFA}
	A Determinisitic Finite (Semi-) Automaton (DFA) $D$ is a triple $(Q, \Sigma, \delta)$ where $Q=[n]$ is a finite set of $n$ states with $1\in Q$ being the source state, $\Sigma=[\sigma]$ is the finite alphabet of size $\sigma$, and $\delta:Q\times\Sigma\rightarrow Q$ is a transition function containing $m$ transitions.
\end{definition}

We omit to specify the final states of DFAs, since they do not play a role in the context of our problem. 
We use the shorthand $\delta_j(v)$ for $\delta(v, j)$. Furthermore, we write $\delta^{out}(v):=\{\delta_j(v): j\in \Sigma\}$ for the set of all out-neighbors of a state $v\in Q$ and $\delta^{in}(v):=\{u\in Q: \exists j\in \Sigma \text{ with }v\in \delta_j(u)\}$ for the set of all in-neighbors of $v$. We assume DFAs to have non-zero in-degree for exactly the non-source states, i.e., $\delta^{in}(v)\neq \emptyset$ if and only if $v>1$;
This choice simplifies our exposition and it is not restrictive from the point of view of the languages accepted by such DFAs.
We do not require the transition function $\delta$ to be complete; This choice is motivated by the fact that requiring completeness restricts the class of Wheeler DFAs \cite{ALANKO2021104820}. Furthermore, we do not require DFAs to be connected; Also this choice is customary as it allows, for instance, to use our WDFA sampler to empirically study properties such as connectivity phase transition thresholds. 

We say that the alphabet $\Sigma$ is \emph{effective} if and only if $(\forall j\in \Sigma)(\exists u,v\in Q)(\delta_j(u)=v)$, i.e. if every character of $\Sigma$ labels at least one transition. 
We assume that the alphabet $\Sigma=[\sigma]$ is totally ordered according to the standard order among integers. 
Wheeler DFAs constitute a special class of DFAs that can be stored compactly and indexed efficiently due to an underlying order on the states: the \emph{Wheeler order} (see Definition \ref{def:WDFA}). As said in Definition \ref{def:DFA}, in this paper the states $Q$ of an automaton $D$ are represented by the integer set $[n]$ for some positive integer $n$; note that in the following definition we use the order on integers $<$ to denote the Wheeler order on the states. 

\begin{definition}[Wheeler DFA~\cite{gagie:tcs17:wheeler}]\label{def:WDFA}
	A \emph{Wheeler DFA} (WDFA)
	is a DFA $D$ such that $<$ is a \emph{Wheeler order}, i.e. for $a,a'\in \Sigma$, $u,v,u',v'\in Q$: 
	\begin{enumerate}[(i)]
	    \item \label{wheeler 1} If $u' = \delta_a(u)$, $v' = \delta_{a'}(v)$, and $a \prec a'$, then $u' < v'$.
	    \item \label{wheeler 2} If $u' = \delta_a(u)\neq \delta_a(v) = v'$ and $u < v$, then $u' < v'$.
	\end{enumerate}
\end{definition}

We note that the source axiom present in \cite{gagie:tcs17:wheeler}, which requires that the source state is first in the order, vanishes in our case as the ordering $<$ on the integers directly implies that the source state is ordered first. Notice that property (\ref{wheeler 1}) in Definition~\ref{def:WDFA} implies that a WDFA is \emph{input-consistent}, i.e., all in-going transitions to a given state have the same label.

\begin{definition}\label{def:D n,m,sigma}
With $\mathcal D_{n, m, \sigma}$ we denote the set of all Wheeler DFAs with effective alphabet $\Sigma=[\sigma]$, $n$ states $Q=[n]$, $m$ transitions, and Wheeler order $1 < 2 < \ldots < n$.     
\end{definition}
 
Clearly, $\mathcal D_{n, m, \sigma}$ is a subset of the set $\mathcal A_{n, m, \sigma}$ of all finite (possibly non-deterministic) automata over the ordered alphabet $[\sigma]$ with $n$ states $[n]$ and $m$ transitions. 

In this paper we investigate the following algorithmic problem: 

\begin{problem}\label{main problem}
For given $n$, $m$, and $\sigma$, generate an element from $\mathcal D_{n, m, \sigma}$ uniformly at random. 
\end{problem}

Note that, since in Definition \ref{def:D n,m,sigma} we require $1 < 2 < \ldots < n$ to be the Wheeler order, Problem \ref{main problem} is equivalent to that of uniformly generating pairs formed by a Wheeler DFA $D$ and a valid Wheeler order for the states $Q=[n]$ of $D$, not necessarily equal to the integer order $1<2<\dots < n$.
Throughout the whole paper, we assume that $n-1 \le m \le n\sigma$ and $\sigma \le n - 1$ (due to input consistency), as otherwise $\mathcal D_{n, m, \sigma}=\emptyset$ and the problem is trivial.

\section{An Algorithm for Uniformly Generating WDFAs}\label{sec:uniform WDFAs}

Our strategy towards solving Problem \ref{main problem} efficiently is to associate every element $D$ from $\mathcal D_{n, m, \sigma}$ to exactly one pair $(O, I)$ of elements from  $\mathcal O_{n, \sigma, m} \times \mathcal I_{m, n}$ (see Definition \ref{def:I O} below) via a function $r: \mathcal D_{n, m, \sigma} \rightarrow \mathcal{O}_{n, \sigma, m} \times \mathcal{I}_{m, n}$ (``$r$'' stands for \emph{representation}). Formally, the two sets appearing in the co-domain of $r$ are given in the following definition.

\begin{definition}\label{def:I O}
Let
\begin{align*}
    &\mathcal{O}_{n, \sigma, m} 
    := \big\{ O \in \{0,1\}^{n\times \sigma} : \|O\|=m \text{ and } \|O_j\| \ge 1 \text{ for all } j\in [\sigma]\big\} \quad\text{ and }\\
    &\mathcal{I}_{m, n}
    := \big\{ I \in \{0,1\}^{m} : \|I\| = n - 1\big\}.
\end{align*}    
\end{definition}

The intuition behind the two sets $\mathcal{O}_{n, \sigma, m}$ and $\mathcal{I}_{m, n}$ is straightforward: their elements encode the outgoing labels and the in-degrees of the automaton's states, respectively. In order to describe more precisely this intuition, let us fix an automaton $D = (Q, \delta, \Sigma) \in \mathcal D_{n, m, \sigma}$ and consider its image $r(D) = (O, I)\in \mathcal{O}_{n, \sigma, m} \times \mathcal{I}_{m, n}$ (see Figures~\ref{fig:running ex D} and \ref{fig:running ex OI} for an illustration):
\begin{itemize}
    \item The matrix $O$ is an encoding of the labels of the out-transitions of $D$. A 1-bit in position $O_{u, j}$ means that there is an out-going transition from state $u$ labeled $j$. 
    Formally, 
    \begin{align}\label{formula: O}
        O_{u, j} := 
        \begin{cases}
            1 &\text{ if }\exists v: v=\delta_{j}(u) \\
            0 &\text{ otherwise.}
        \end{cases}
    \end{align}

    \item The vector $I$ is a concise encoding of the in-degrees of all states. It is defined as 
    \begin{align}\label{formula: I}
        I:=(
            \underbrace{1, 0, \ldots, 0}_{|\delta^{in}(2)|}, 
            \underbrace{1, 0, \ldots, 0}_{|\delta^{in}(3)|}, \ldots, 
            \underbrace{1, 0, \ldots, 0}_{|\delta^{in}(n)|}), 
    \end{align}
    i.e, for all states $i$ other than the source (that has no in-transitions), the vector contains exactly one 1-bit followed by $|\delta^{in}(i)| - 1$ 0-bits. 
\end{itemize}
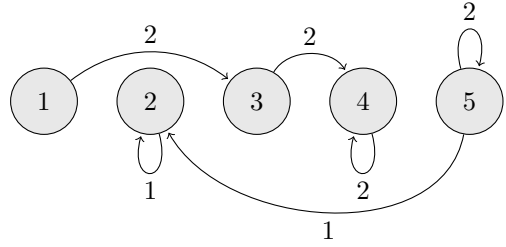
\begin{figure}[ht!]
	\centering
	\begin{tikzpicture}[shorten >=1pt,node distance=1.4cm,on grid,auto]
	\tikzstyle{every state}=[fill={rgb:black,1;white,10}]
	
	\node[state] (1)                  {$1$};
	\node[state] (2)  [right of=1]    {$2$};
 	\node[state] (3)  [right of=2]    {$3$};
	\node[state] (4)  [right of=3]    {$4$};
	\node[state] (5)  [right of=4]    {$5$};

	\path[->]
	(1) edge [bend left = 38] node {2}    (3)
        (2) edge [loop below] node {1}    (2)
        (4) edge [loop below] node {2}    (4)
        (5) edge [loop above] node {2}    (5)
        (5) edge [bend left = 80, in = 120] node {1}    (2)
        (3) edge [bend left = 60] node {2}    (4);

	\end{tikzpicture}\caption{Running example: a WDFA $D$ with $n=5$ states, $m=6$ edges, alphabet cardinality $\sigma=2$, and Wheeler order $1<2<3<4<5$. Note that the WDFA has two connected components.}\label{fig:running ex D}
\end{figure}
\begin{figure}[ht!]
\centering
\begin{tabular}{cc}
    \begin{minipage}{.2\linewidth}
        \begin{tabularx}{50pt}{ccc}
              & 1 & 2\\
              \cline{2-3}
              1 & \multicolumn{1}{|c|}{0} & \multicolumn{1}{c|}{1} \\
              \cline{2-3}
              2 & \multicolumn{1}{|c|}{1} & \multicolumn{1}{c|}{0} \\
              \cline{2-3}
              3 & \multicolumn{1}{|c|}{0} & \multicolumn{1}{c|}{1} \\
              \cline{2-3}
              4 & \multicolumn{1}{|c|}{0} & \multicolumn{1}{c|}{1} \\
              \cline{2-3}
              5 & \multicolumn{1}{|c|}{1} & \multicolumn{1}{c|}{1} \\
              \cline{2-3}
        \end{tabularx}
    \end{minipage} &
    \begin{minipage}{.2\linewidth}
        \begin{tabularx}{100pt}{cccccc}
        2 & & 3 & 4 & & 5\\
        \cline{1-6}
             \multicolumn{1}{|c}{\textbf{1}} & \multicolumn{1}{|c}{0} & \multicolumn{1}{|c}{\textbf{1}} & \multicolumn{1}{|c}{1} & \multicolumn{1}{|c}{0} & \multicolumn{1}{|c|}{1} \\
        \cline{1-6}
        \end{tabularx}
    \end{minipage} 
\end{tabular}\caption{Matrix $O$ (left) and bit-vector $I$ (right) forming the encoding $r(D)=(O,I)$ of the WDFA $D$ of Figure \ref{fig:running ex D}. In matrix $O$, column names are characters from $\Sigma=[\sigma]$ and row names are states from $Q=[n]$. In bit-vector $I$, each state (except state 1) is associated with a bit set, in Wheeler order. Cells containing a set bit are named with the name of the corresponding state. Bits in bold highlight the states on which the character that labels the state's incoming transitions changes (i.e. state 2 is the first whose incoming transitions are labeled 1, and state 3 is the first whose incoming transitions are labeled 2).}\label{fig:running ex OI}
\end{figure}
Let us proceed with two remarks. 
\begin{remark}
  As $\|O\|=m$ there are $m$ transitions in total. As $\|O_j\| \ge 1$ for all $j\in [\sigma]$, the alphabet is effective, i.e., every character labels at least one transition. 
\end{remark}

\begin{remark}
 The vector $I$ does not encode the letter on which a transition is in-going to a given state. Notice however that as $D$ is a WDFA all these transitions have to be labeled with the same letter and we can reconstruct this letter for a given $I$ once we know the total number of transitions labeled with each letter. This is because property~(\ref{wheeler 1}) of Definition~\ref{def:WDFA} guarantees that the node order is such that the source state (that has no in-going transitions) is ordered first followed by nodes whose in-transitions are labeled with character $1$, followed by nodes with in-transitions labeled with character $2$, etc. The information on how many transitions are labeled with each character is carried by the matrix $O$ for which $r(D)=(O, I)$.
\end{remark}

Let $(O, I)$ be a pair from the image of $r$, i.e, $r(D)=(O, I)$ for some $D$. Then it will always be the case that $I$ is contained in a subset $\mathcal I_O$ of $\mathcal I_{m,n}$ that can be defined as follows.

\begin{definition}
    For a matrix $O\in \mathcal O_{n, \sigma, m}$, let 
    \[
            \mathcal I_O:= \Big\{ I\in \mathcal I_{m,n}: I_{1 + \sum_{k=1}^{j - 1} \|O_k\|} = 1 \text{ for all } j\in [\sigma] \Big\}.
    \]
\end{definition}

Using our running example of Figure \ref{fig:running ex OI}, the bits $I_{1 + \sum_{k=1}^{j - 1} \|O_k\|}$ that we force to be equal to 1 are those highlighted in bold, i.e. $I_1$ and $I_3$: 
noting that bits in $I$ correspond to edges, 
those bits correspond to the leftmost edge labeled with a given character $j$ (for any $j\in \Sigma$).

This leads us to define the following subset of $\mathcal O_{n, \sigma, m} \times \mathcal I_{m,n}$:

\begin{definition}\label{def:new signature}
    $
    \mathcal R_{n, m, \sigma} := \{(O, I): O\in \mathcal O_{n, \sigma, m}\text{ and }I\in\mathcal I_O\}.
    $
\end{definition}

Based on the above definition, we can prove:

\begin{lemma}
    For any $D \in \mathcal D_{n, m, \sigma}$, $r(D) \in \mathcal R_{n, m, \sigma}$.
\end{lemma}
\begin{proof}
Note that the integers $\sum_{k=1}^{j - 1}\|O_k\|$ for $j\in [\sigma]$ correspond to the number of edges labeled with letters $1, \ldots, j - 1$, hence the positions $1 + \sum_{k=1}^{j - 1} \|O_k\|$ correspond to 
a change of letter in the sorted (by destination node) list of edges. Recalling that WDFAs are input-consistent (i.e., all in-transitions of a given node carry the same label) and that nodes are ordered by their in-transition letters, positions $1 + \sum_{k=1}^{j - 1} \|O_k\|$ for $j\in [\sigma]$ in $I$ must necessarily correspond to the first edge of a node, hence they must contain a set bit.
\end{proof}

The co-domain of the function $r$ can thus be restricted, and the function's signature can be redefined, as follows: $r:\mathcal D_{n, m, \sigma} \rightarrow \mathcal R_{n, m, \sigma}$.

After describing this association of a WDFA $D\in \mathcal D_{n, m, \sigma}$ to a (unique) pair $r(D)=(O, I) \in \mathcal R_{n, m, \sigma}$, 
we will argue that 
function $r$ is indeed a bijection from  $\mathcal D_{n, m, \sigma}$ to $\mathcal R_{n, m, \sigma}$. It will follow that one way of generating elements from $\mathcal D_{n, m, \sigma}$ is to generate elements from $\mathcal R_{n, m, \sigma}$: this will lead us to an efficient algorithm to uniformly sample WDFAs from $\mathcal D_{n, m, \sigma}$, as well to a formula for the cardinality of $\mathcal D_{n, m, \sigma}$.

\subsection{The Basic WDFA Sampler} \label{sec:algo}

Our overall approach is to (1) uniformly sample a matrix $O$ from $\mathcal O_{n, \sigma, m}$ using Algorithm \ref{alg: sample_O}, then (2) uniformly sample a vector $I$ from $\mathcal I_O$ using Algorithm \ref{alg: sample_I} with input $O$, and finally (3) build a WDFA $D$ using $O$ and $I$ as input via Algorithm \ref{alg: build_D}. 
We summarize this procedure in Algorithm \ref{alg: generate_D}.
A crucial point in our correctness analysis (Section \ref{sec:analysis}) will be to show that uniformly sampling from $\mathcal O_{n, \sigma, m}$ and $\mathcal I_O$ does indeed lead to a uniform WDFA from $\mathcal D_{n, m, \sigma}$ (besides the bijectivity of $r$, intuitively, this is because $|\mathcal I_O|=|\mathcal I_{O'}|$ for any $O,O'\in \mathcal O_{n, \sigma, m}$).

As source of randomness, our algorithm uses a black-box \emph{shuffler} algorithm: given a bit-vector $B\in\{0,1\}^*$, function $\mathtt{shuffle}(B)$ returns a random permutation of $B$.
To improve readability, in this subsection we start by describing a preliminary simplified version of our algorithm which does not assume any particular representation for the matrix-bit-vector pair $(O,I) \in \mathcal R_{n, m, \sigma}$, nor a particular shuffling algorithm (for now, we only require the shuffling algorithm to permute uniformly its input). 
By employing a particular \emph{sequential shuffler}, in Subsection \ref{sec:constant space} we then show that 
we can generate a sparse representation of $O$ and $I$ on-the-fly, thereby achieving \emph{constant} working space and linear expected running time.

\begin{algorithm}[!ht]
\caption{\texttt{sample\_D$(n, m , \sigma)$ }}\label{alg: generate_D}
    $O:= \texttt{sample\_O}(n, m , \sigma)$\;
    $I:= \texttt{sample\_I}(O)$\;
    $D:= \texttt{build\_D}(O, I)$\;
    \Return{$D$}
\end{algorithm}

\subparagraph{Out-transition Matrix.}
In order to sample the matrix $O$ from $\mathcal O_{n, \sigma, m}$, 
in addition to function $\texttt{shuffle}$ we assume
a function $\texttt{reshape}_{k, \ell}$ that takes a vector $x$ of dimension $k\cdot \ell$ and outputs a matrix $A$ of dimension $k\times \ell$ with the $j$'th column $A_j$ being the portion $x_{(j - 1)\cdot k + 1}, \ldots, x_{j\cdot k}$
of $x$. The algorithm to uniformly generate $O$ from $\mathcal O_{n, \sigma, m}$ then simply samples a bit vector of length $n\sigma$ with exactly $m$ 1-bits, shuffles it uniformly, reshapes it to be a matrix of dimension $n\times \sigma$ and repeats these steps until a matrix is found with at least one 1-bit in each column (rejection sampling).

\begin{algorithm}[!ht]
\caption{\texttt{sample\_O$(n, m , \sigma)$}}\label{alg: sample_O}
    \Repeat{
        $\|O_j\|\ge 1$ for all $j\in [\sigma]$
    }{
        $O:= \texttt{reshape}_{n, \sigma}(\texttt{shuffle}(1^m 0^{n\sigma - m}))$\;
    }
    \Return{$O$}
\end{algorithm}

Looking at the running example of Figures \ref{fig:running ex D} and \ref{fig:running ex OI}, the shuffler is called as $\texttt{shuffle}(1^60^4)$. In this particular example, this bit-sequence is permuted as 0100110111 by function $\texttt{shuffle}$. Function $\texttt{reshape}_{n, \sigma}$ converts this bit-sequence into the matrix $O$ depicted in Figure \ref{fig:running ex OI}, left. 

\subparagraph{In-transition Vector.}
In order to generate the vector $I$ from $\mathcal I_{m, n}$, we proceed as follows. The algorithm takes $O$ as input and generates a uniform random element from the set $\mathcal I_O$
by first creating a ``mask'' that is a vector of the correct length $m$ and contains $\sigma$ 1-bits at the points $1 + \sum_{k=1}^{j - 1} \|O_k\|$ for $j\in[\sigma]$. These are the points in $I$ where the character of the corresponding transition changes and hence, by the input-consistency condition, also the state has to change. The remaining $m - \sigma$ positions in the mask are filled with the wildcard character $\#$. We then give this mask vector as the first argument to a function \texttt{fill} that replaces the $m - \sigma$ positions that contain the wildcard character $\#$ with the characters in the second argument (in order). Formally, the function \texttt{fill} takes two vectors as arguments $a$ and $b$ with the condition that $a$ contains $|b|$ times the $\#$ character and $|a| - |b|$ times a 1-bit. The function then returns a vector $c$ that satisfies $c_i = 1$ whenever $a_i = 1$ and $c_i = b_{i - \rank(a, i)}$ otherwise, i.e., when $a_i=\#$. 
\begin{algorithm}[!ht]
\caption{\texttt{sample\_I$(O)$}}\label{alg: sample_I}
    extract $n, m, \sigma$ from $O$\;
    $\mask := 1\#^{\|O_1\| - 1}1\#^{\|O_2\| - 1}\ldots 1\#^{\|O_\sigma\| - 1} $\;
    $I := \texttt{fill}(\mask, \texttt{shuffle}(1^{n - \sigma - 1} 0^{m - n + 1}))$\;
    \Return{$I$}
\end{algorithm}

Going back to our running example of Figures \ref{fig:running ex D} and \ref{fig:running ex OI}, we have mask = \textbf{1}\#\textbf{1}\#\#\# (that is, all bits but the bold ones in the right part of Figure \ref{fig:running ex OI} are masked with a wildcard). The shuffler is called as \texttt{shuffle}(1100) and, in this particular example, returns the shuffled bit-vector 0101. Finally, function \texttt{fill} is called as \texttt{fill}(\textbf{1}\#\textbf{1}\#\#\#,0101) and returns the bit-vector $I=$101101 depicted in the right part of Figure \ref{fig:running ex OI}.

\subparagraph{Building the WDFA.}
After sampling $O$ and $I$, the remaining step is to build the output DFA $D$. This is formalized in Algorithm~\ref{alg: build_D}. By iterating over all non-zero elements in $O$, we construct the transition function $\delta$: the $i$'th non-zero entry in $O$ corresponds to an in-transition at state $\rank(I, i) + 1$ (we keep a counter named $v$ corresponding to this rank). The origin state of the transition is the row in which we find the $i$'th 1 in $O$ when reading $O$ column-wise. The column itself corresponds to the label of this transition.
\begin{algorithm}[!ht]
\caption{\texttt{build\_D$(O, I)$}}\label{alg: build_D}
    extract $n, m, \sigma$ from $O$, $Q:= [n]$, $\Sigma:=[\sigma]$\;
    \medskip
    
    $\delta := \emptyset$, 
    $i := 1$, $v := 1$\;
    
    \For{$j = 1,2,\dots,\sigma$}{
        \For{$u = 1,2, \dots, n$}{
            \If{$O_{u, j}=1$}{
                \lIf{$I_i=1$}{
                    $v := v + 1$
                }\label{line: if I}
                $\delta := \delta \cup \{((u, j), v))\}$, $i := i + 1$\;\label{line: add edge}
            }
        }
    }
    \Return{$D=(Q, \Sigma, \delta)$}
\end{algorithm}

\subsection{Constant-Space WDFA Sampler}\label{sec:constant space}

Notice that our Algorithm \ref{alg: build_D} accesses the matrix $O$ and the bit-vector $I$ in a sequential fashion: $O$ is accessed column-wise 
and $I$ from its first to last position. Based on this observation, we now show how our WDFA sampler can be modified to use \emph{constant} working space. The high-level idea is to generate on-the-fly the positions of non-zero entries of $O$ and $I$ in increasing order. 

In order to achieve this, we employ the 
\emph{sequential shuffler} described by Shekelyan and Cormode~\cite{hiddenShuffle}. Given two integers $N$ and $n$, the function $\mathtt{init\_sequential\_shuffler}(N, n)$ returns an iterator $S$ that can be used (with a stack-like interface) to extract $n$ uniform integers without replacement from $[N]$, in \emph{ascending order} and using a \emph{constant} number of words of working space (that is, the random integers are generated on-the-fly upon request, from the smallest to the largest). More specifically, function $S.\mathtt{pop}()$ returns the next sampled integer, while $S.\mathtt{empty}()$ returns true if and only if all $n$ integers have been extracted. The sequential shuffling algorithm is essentially a clever modification of Knuth's shuffle~\cite{Knuth98} (also referred to as Fisher-Yates shuffler). Knuth's shuffler, after going through the arbitrarily ordered set $[N]$, and in the $i$'th iteration (for $i$ from $1$ to $n$) swapping the $i$'th item with the item at a random position $[i, N]$, returns the first $n$ items in the resulting permutation. Knuth's shuffler requires working space proportional to $n$ as we need to remember which elements have been swapped from lower positions (i.e., index $\le n$) into higher positions (i.e., index $>n$). The idea behind the sequential shuffler of Shekelyan and Cormode is to first sample just the cardinality $H$ of the set of items in higher positions that Knuth’s shuffler would swap into lower position. Then, in a second step the algorithm samples $H$ actual items from higher positions with replacement, resulting in $h\le H$ elements. Finally, in a third step, $n - h$ items are sampled from lower positions. We note that the distribution in the first step is chosen such that the sampling in the second step can be done with replacement -- sampling duplicates simply increases the number of items sampled from lower positions. We refer the reader to the article by Shekelyan and Cormode~\cite{hiddenShuffle} for further details.


\begin{algorithm}[!ht]
\caption{\texttt{sample\_D\_constant\_space$(n,m,\sigma)$}}\label{alg: constant space}
    $i := 1$\tcc*[r]{\scriptsize Current position in the subsequence of \#'s of the mask}
    $v := 1$\tcc*[r]{\scriptsize Destination state of current transition} 

    \medskip

    $S_O := \mathtt{init\_sequential\_shuffler}(n\sigma, m)$\;\label{line: init SO}
    $S_I := \mathtt{init\_sequential\_shuffler}(m-\sigma, n-\sigma-1)$\;\label{line: init SI}

    \medskip
    
    $i' := S_I.\mathtt{pop}()$\tcc*[r]{\scriptsize next nonzero position in sequence of \#'s in the mask}
    $j := 0$\tcc*[r]{\scriptsize current column in $O$}
    $prev\_j := 0$\tcc*[r]{\scriptsize previously-visited column in $O$}
    
    \medskip

    \While{$\mathrm{not}\ S_O.\mathtt{empty}()$\label{line:while}}{

        $t := S_O.\mathtt{pop}()$\;\label{line:extract t}
        $(u, j) := \Big(\big((t-1)\mod n\big) + 1,\big((t-1) \  \mathrm{div}\ n\big) + 1\Big)$\tcc*[r]{\scriptsize Nonzero coordinate in $O$} \label{line: nonzero O}

        \If{$j > prev\_j+1$}{
        clear output stream and goto line 1\tcc*[r]{\scriptsize Rejection: $\|O_{prev\_j+1}\|=0$}\label{line: reject 1}
        }

        \eIf{$j = prev\_j+1$\tcc*[r]{\scriptsize Column of $O$ changes}\label{line: if col changes}}{
            $v := v + 1$\;\label{line: incr v, col changes}
            $prev\_j := j$\;
        }{

            \If{$i = i'$}{
                $v := v + 1$\;\label{line: incr v 2}
                $i' := S_I.\mathtt{pop}()$\tcc*[r]{\scriptsize next nonzero position in sequence of \#'s in the mask}\label{line: extract i'}
            }           
            $i := i + 1$\; \label{line:incr i}
        }
        \textbf{output} $((u, j), v)$\tcc*[r]{\scriptsize Stream transition to output}\label{line: output} 
    }
    \If{$j \neq \sigma$}{
        clear output stream and goto line 1\tcc*[r]{\scriptsize Rejection: $\|O_{\sigma}\|=0$\label{line:rejection 2}}
    }
\end{algorithm}

\subparagraph{Algorithm Description.}
We now describe Algorithm \ref{alg: constant space}. 
We recall the mask employed in Algorithm \ref{alg: sample_I}:  
Algorithm \ref{alg: constant space} 
iterates, using variable $i$, over the ranks (i.e., $i$-th occurrences) of characters \# (wildcards) in the mask. 
Variable $i'$, on the other hand, stores the rank of the next wildcard \# that is replaced with a set bit by the shuffler; the values of $i'$ are extracted from the shuffler $S_I$. 
Now, whenever $i=i'$, we are looking at a bit set in bit-vector $I$ (which here is not stored explicitly, unlike in Algorithm~\ref{alg: build_D}) and thus we have to move to the next destination state $v$. This procedure exactly simulates Lines~\ref{line: if I} and \ref{line: add edge} of Algorithm \ref{alg: build_D}. 

The iteration (column-wise) over all non-zero entries of matrix $O$ is simulated by the extraction of values from the shuffler $I_O$ (one value per iteration of the while loop at Line \ref{line:while}): each such value $t$ extracted at Line \ref{line:extract t} is converted to a pair $(u,j)$ at Line \ref{line: nonzero O}. Variables $j$ and $prev\_j$ store the columns of the current and previously-extracted non-zero entries of $O$, respectively. If $j>prev\_j+1$, then column number $prev\_j+1$ has been skipped by the shuffler, i.e., $O_{prev\_j+1}$ does not contain non-zero entries. In this case, we reject and start the sampler from scratch (Line~\ref{line: reject 1}; note that we need to clear the output stream before re-initializing the algorithm). If, on the other hand, $j=prev\_j+1$ (Line~\ref{line: if col changes}), then the current non-zero entry of $O$ belongs to the next column with respect to the previously-extracted non-zero entry; this means that the character labeling incoming transitions changes and we therefore move to the next destination node by increasing $v:=v+1$ (Line~\ref{line: incr v, col changes}). In this case we do not increment $i$, since the new destination node $v$ is the first having incoming label $j$ and thus it does not correspond to a character \# in the mask.
Variable $i$ gets incremented only if $j=prev\_j$: this happens at Line \ref{line:incr i}.
The other case in which we need to move to the next destination node ($v:=v+1$) is when $j = prev\_j$ and $i=i'$ (Line \ref{line: incr v 2}). In such a case, in addition to incrementing $v$ we also need to extract from the shuffler $S_I$ the rank $i'$ of the next mask character \# that is replaced with a set bit (Line \ref{line: extract i'}). After all these operations, we write the current transition $((u,j),v)$ to the output stream (Line \ref{line: output}). 
The last two lines of Algorithm \ref{alg: constant space} check if the last visited column of matrix $O$ is indeed $O_{\sigma}$. If not, $\|O_{\sigma}\|=0$ and we need to reject and re-start the algorithm.

The remaining components of Algorithm~\ref{alg: constant space} are devoted to simulate
Algorithm~\ref{alg: generate_D}, using as input the two sequences of random pairs/integers extracted from $S_O$ and $S_I$, respectively. As a matter of fact, the two loops in Algorithm~\ref{alg: build_D} correspond precisely to extracting the pairs $(u,j)$ from $S_O$, and the check at Line~\ref{line: if I} of Algorithm \ref{alg: build_D}, together with the increment of $i$ at Line~\ref{line: add edge}, corresponds to extracting the integers $i'$ from $S_I$. The rejection sampling mechanism (repeat-until loop in Algorithm~\ref{alg: sample_O}) is simulated in Algorithm~\ref{alg: constant space}
by re-starting the algorithm whenever the column $j$ of the current pair $(u,j)$ is either larger by more than one unit than the column $j\_prev$ of the previously-extracted pair (i.e., $\|O_{j\_prev+1}\|=0$, Line~\ref{line: reject 1}), or if the last pair extracted from $S_O$ is such that $j$ is not the $\sigma$-th column (i.e., $\|O_{\sigma}\|=0$, Line~\ref{line:rejection 2}).

\subparagraph{Running Example.}
To understand how the sequential shuffler is used in Algorithm \ref{alg: constant space}, refer again to the running example of Figures \ref{fig:running ex D} and \ref{fig:running ex OI}. In Algorithm \ref{alg: constant space} at Line \ref{line: init SO},
the sequential shuffler $S_O$ is initialized as $S_O := \mathtt{init\_sequential\_shuffler}(n\sigma = 10, m=6)$, i.e. the iterator $S_O$ returns 6 uniform integers without replacement from the set $\{1,2,\dots, 10\}$. In this particular example, 
function $\texttt{pop}()$ called on iterator $S_O$ returns the following integers, in this order: 2,5,6,8,9,10. Using the formula at Line \ref{line: nonzero O} of Algorithm \ref{alg: constant space}, these integers are converted to the matrix coordinates $(2,1), (5,1), (1,2), (3,2), (4,2), (5,2)$, i.e., precisely the nonzero coordinates of matrix $O$ in Figure \ref{fig:running ex OI}, sorted first by column and then by row. 

Using the same running example, the sequential shuffler $S_I$ is initialized in Line \ref{line: init SI} of Algorithm \ref{alg: constant space} as $S_I := \mathtt{init\_sequential\_shuffler}(m-\sigma = 4,n-\sigma - 1 = 2)$, i.e. the iterator $S_I$ returns two  uniform integers without replacement from the set $\{1,2,3,4\}$. In this particular example, 
function $\texttt{pop}()$ called on iterator $S_I$ returns the following integers, in this order: 2,4. Using the notation of the previous subsection, this sequence has the following interpretation: the 2-nd and 4-th occurrences of \# of our mask 1\#1\#\#\# used in Algorithm \ref{alg: sample_I} have to be replaced with a bit 1, while the others with a bit 0. After this replacement, the mask becomes 101101, i.e. precisely bit-vector $I$ of Figure \ref{fig:running ex OI}.

\section{Analysis}\label{sec:analysis}

\subsection{Correctness, Completeness and Uniformity} 

Being Algorithm \ref{alg: constant space} functionally equivalent to Algorithm \ref{alg: generate_D} (the only relevant difference between the two being the employed data structures to represent matrix $O$ and bit-vector $I$), for ease of explanation in this section we focus on analyzing the 
correctness (the algorithm generates only elements from $\mathcal D_{n, m, \sigma}$), completeness (any element from $\mathcal D_{n, m, \sigma}$ can be generated by the algorithm) and uniformity (all $D\in\mathcal D_{n, m, \sigma}$ have the same probability to be generated by the algorithm) of  Algorithm \ref{alg: generate_D}.
These properties then automatically hold on Algorithm \ref{alg: constant space} as well.

We start with a simple lemma. The lemma says the following: Assume that $r(D)=(O, I)$ and $O$ contains a 1 in position $(u,j)$, meaning that there is a transition leaving state $u$, labeled with letter $j$. Then this out-transition is the $i=\rank(O, (u, j))$ bit that is set to 1 in $O$ and hence the entry in $I$ corresponding to this transition can be found at $I[i]$. The state to which this transition is in-going is exactly the number of 1s in $I$ up to this point, i.e., $\rank(I, i)$, plus one (the offset is due to the source having no in-transitions).
\begin{lemma}
\label{lemma: O I relation}
    Let $D\in\mathcal D_{n, m, \sigma}$ and let $r(D)=(O, I)$. If $O_{u, j} = 1$ and $\rank(O, (u, j))=i$, then $\delta(u, j)=\rank(I, i) + 1$.
\end{lemma}
\begin{proof}
    First, notice that since $O_{u,j}=1$, it is clear that there is an outgoing transition from $u$ labeled $j$. Furthermore, since $\rank(O, (u, j))=i$, we know that this transition corresponds to the $i$-th entry in $I$. Now, by the definition of $I$, it follows that the destination state of the considered transition is $v=\rank(I,i)+1$.
\end{proof}

Algorithm~\ref{alg: build_D} is a deterministic algorithm and thus describes a function, say $f$, from the set of its possible inputs to the set of its possible outputs. The set of its possible inputs, i.e., the domain of $f$, is exactly $\mathcal R_{n, m, \sigma}$. The algorithm's output is certainly a finite automaton, i.e., the co-domain of $f$ is $\mathcal A_{n, m, \sigma}$. We will in fact show that the range of $f$ is exactly $\mathcal D_{n, m, \sigma}$. We will do so by showing that $f$ is actually an inverse of $r$, more precisely we show that (1) $r$ is surjective and (2) $f$ is a left-inverse of $r$ (and thus $r$ is injective).

\subparagraph{Surjectivity of $r$.} 
We start with proving that $r$ is surjective. 
\begin{lemma}
    It holds that $r:\mathcal D_{n, m, \sigma} \rightarrow \mathcal R_{n, m, \sigma}$ is surjective.
\end{lemma}
\begin{proof}
    Fix an element $(O, I)\in \mathcal R_{n, m, \sigma}$, i.e., an $O\in \mathcal O_{n, \sigma, m}$ and $I\in\mathcal I_O$. We now construct an automaton $D=(Q, \Sigma, \delta)$ and then show that $r(D)=(O, I)$. We let $Q=[n]$, $\Sigma=[\sigma]$, and
    \[
        \delta = \{((u, j), v): O_{u, j} = 1 \text{ and } v = \rank(I, \rank(O, (u, j))) + 1\}.
    \]
    Let $r(D)=(O', I')$ and let us proceed by showing that $O=O'$ and $I=I'$. Recall the definition of $r$, see Equations~\eqref{formula: O} and~\eqref{formula: I}. It is immediate that $O'=O$ given the definition of $O'$ and $\delta$. In order to show that $I'=I$, first note that $I'=\prod_{i=2}^n10^{|\delta^{in}(i)| - 1}$. Then, consider the following relation between $I$ and $\delta^{in}(i)$ for any state $i\in[n]$, which uses the definition of $\delta$ and Lemma~\ref{lemma: O I relation}:
    \begin{align*}
        |\delta^{in}(i)|
        &= |\{(u, j)\in [n]\times[\sigma] : O_{u,j} = 1 \text{ and } \rank(I, \rank(O, (u, j))) + 1 = i\}|\\
        &= |\{k\in [m] : k = \rank(O, (u, j)) \text{ for some }(u,j)\in [n]\times[\sigma] \text{ with }O_{u,j}=1\\ 
            &\hspace{8cm} \text{ and } \rank(I, k) + 1 = i\}|\\
        & = \max\{k\in [m] : \rank(I, k)=i-1\}
            - \max\{k\in [m] : \rank(I, k)=i-2\}.
    \end{align*}
    Now, recall that $I\in \mathcal{I}_O$, hence the first bit in $I$ is 1. Using the previous equality, it now follows that the second 1-bit in $I$ is at position $|\delta^{in}(2)| + 1$. By using this equality another $n-3$ times, we obtain that the first $\sum_{i=2}^{n - 1}|\delta^{in}(i)| + 1$ positions of $I$ are equal to $(10^{|\delta^{in}(2)| - 1}10^{|\delta^{in}(3)| - 1}\ldots 10^{|\delta^{in}(n-1)| - 1}) 1$ and thus agree with $I'$. It remains to observe that this portion of $I$ already contains $n-1$ bits that are equal to 1 and thus the remaining bits have to be zero-bits as $I\in \mathcal I$. Hence, $I = I'$ and this completes the proof.
\end{proof}

\subparagraph{Injectivity of $r$ via Left-inverse $f$.}
In order to establish that $f$ is the inverse of $r$, it remains to prove that $f$ is a left-inverse of $r$ (which implies that $r$ is injective).
\begin{lemma}
    The function $f$ is a left-inverse of $r$, i.e.,  $f(r(D))=D$ for any $D\in \mathcal D_{n, m, \sigma}$. 
\end{lemma}
\begin{proof}
    Let $D=(Q, \Sigma, \delta)\in \mathcal D_{n, m, \sigma}$ and let $(O, I) = r(D)$. We have to show that $D'=f(O, I)$, i.e., the automaton $D'=(Q', \Sigma', \delta')$ output by Algorithm~\ref{alg: build_D} on input $(n, m, \sigma, O, I)$ is equal to $D$. Notice that clearly $Q=Q'=[n]$ and $\Sigma = [\sigma]$. It remains to show that $\delta=\delta'$. It is clear that Algorithm~\ref{alg: build_D} adds $m$ transitions to $\delta'$, one in each of the $m=\|O\|$ iterations. It thus remains to prove that each such transition $((u, j), v)$ added in some iteration $i$ is contained in $\delta$. Firstly, as $O_{u, j}=1$ it is clear that $D$ has an outgoing transition at state $u$ with letter $j$, second it is clear that the algorithm maintains the property that $v=\rank(I, i) + 1$ and thus due to Lemma~\ref{lemma: O I relation} it holds that $\delta(u, j)=v$ and thus this transition is also contained in $\delta$.
\end{proof}

We can thus denote the function $f$ with $r^{-1}$.

\begin{corollary}\label{cor:r is bijective}
    Function $r:\mathcal D_{n, m, \sigma} \rightarrow \mathcal R_{n, m, \sigma}$ is bijective.
\end{corollary}

The above lemma has several consequences. First, it shows that the output of Algorithm~\ref{alg: build_D} is always a WDFA. 
Second, as the function $r$ is bijective, this 
means that generating uniform pairs from the range of $r$ results in a uniform distribution of WDFAs from $\mathcal D_{n, m, \sigma}$.

\begin{lemma}\label{lem: uniform}
    Algorithm~\ref{alg: generate_D} on input $n, m, \sigma$ generates uniformly distributed WDFAs from $\mathcal D_{n, m, \sigma}$.
\end{lemma}
\begin{proof}
    In the light of $r$ being a bijection and Algorithm \ref{alg: build_D} implementing the function $r^{-1}$, it remains to argue that the statements $O:= \texttt{sample\_O}(n, m , \sigma)$ and 
    $I:= \texttt{sample\_I}(O)$ from Algorithm~\ref{alg: generate_D} in fact generate uniformly distributed pairs from the domain of $r^{-1}$, i.e., from $\mathcal R_{n, m, \sigma}$. It is clear that \texttt{sample\_O}$(n, m , \sigma)$ results in a uniformly distributed element $O$ from $\mathcal O_{n, \sigma, m}$ and that \texttt{sample\_I}$(O)$ results in a uniformly distributed element $I$ from $\mathcal I_O$. It thus remains to observe that $|\mathcal I_{O}|$ is identical for all $O\in \mathcal O_{n, \sigma, m}$, namely $|\mathcal I_{O}| = \binom{m - \sigma}{n - \sigma -1}$ for all $O\in \mathcal O_{n, \sigma, m}$. This completes the proof.
\end{proof}

\subsection{Run-time and Space}

We now analyze the number of iterations of Algorithm~\ref{alg: sample_O}, that is, the expected number of rejections before extracting a bit-matrix $O$ with $\|O_j\|>0$ for all $j\in[\sigma]$. Algorithm~\ref{alg: constant space} is clearly equivalent to Algorithm \ref{alg: generate_D} also under this aspect, since at Lines \ref{line: reject 1} and \ref{line:rejection 2} we re-start the algorithm whenever we generate a column $O_j$ without non-zero entries. We prove:

\begin{restatable}{lemma}{analysis}\label{lem: rejection rate}
    Assume that $m\ge \sigma \ln (e \cdot \sigma)$. The expected number of iterations of Algorithm~\ref{alg: sample_O} (equivalently, rejections of Algorithm \ref{alg: constant space}) is at most $1.6$. Furthermore, the algorithm terminates after $O(\log m)$ iterations with probability at least $1 - m^{-c}$ for any constant $c>0$. 
\end{restatable}
\ArxivOrCr{
\begin{proof}
    Consider a fixed iteration of the repeat loop and a fixed column $j\in [\sigma]$ of the matrix $O$. Our goal is to bound the probability that $\|O_j\|=0$. Let us denote this event with $E_j$. Consider the following urn experiment, equivalent to our setting. An urn contains $n$ white and $(\sigma - 1)\cdot n$ black balls, we draw $m$ balls without replacement (in our case, the $n$ white balls correspond to the $n$ cells of $O_j$). Let $F_i$, for $i\in [m]$ denote the event that we do not draw a white ball in the $i$'th draw. It holds that $\Pr[E_j]=\Pr[\bigcap_{i\in [m]}F_i]$, the probability that we did not draw a white ball in any of the $m$ trials (i.e.\ this is the probability that none of the cells of $O_j$ is selected and gets a 1-bit, after inserting $m$ 1-bits in $O$ at random positions without replacement).
    This probability can be expressed as 
    \begin{align*}
        \Pr\Big[\bigcap_{i\in [m]}F_i\Big] 
        &= \prod_{i\in [m]} \Pr\Big[F_i \;\big|\; \bigcap_{k = 1}^{i - 1} F_k \Big]
        = \prod_{i\in [m]} \frac{(\sigma - 1) n - (i - 1)}{\sigma n - (i - 1)}\\
        &= \prod_{i\in [m]} \Big( 1 - \frac{n}{n\sigma - (i - 1)}\Big)
        \le \Big( 1 - \frac{1}{\sigma }\Big)^{m}
        \le \Big( 1 - \frac{1}{\sigma }\Big)^{\sigma \ln (e\cdot \sigma)}
        \le \frac{1}{e\cdot\sigma}.
    \end{align*}
    Now using a union bound over all $\sigma$ columns, we get that the probability that there exists $j\in[\sigma]$ such that $\|O_j\|=0$ is upper bounded by $1/e$. It follows that the expected number of repetitions until it holds that $\|O_j\|=0$ for all $j\in[\sigma]$ is $1 / (1 - 1 / e) \le 1.6$. This shows the bound on the expected number of iterations. Now fix some constant $c>0$ and assume that the algorithm's repeat loop runs for $T = c\ln m$ iterations. We then get that the probability that in each of these loop iterations, there exists a column $j\in[\sigma]$ such that $\|O_j\|=0$, is at most $(1/e)^{c\ln m} = m^{-c}$. Hence, the algorithm terminates after at most $c\ln m$ iterations with probability at least $1 - m^{-c}$.
\end{proof}
}{We refer the reader to the full version of this article~\cite{full-version} for the proof of the above lemma.}
Now assume that $\sigma\le m / \ln m$. This implies that $e \cdot \sigma \le m$ (for $m$ larger than a constant), which together with the initial assumption implies that $\sigma \ln (e \cdot \sigma) \le \sigma \ln m \le m$. This is exactly the condition in 
Lemma~\ref{lem: rejection rate}. 
Hence, if $\sigma \leq m/\ln m$ then the expected number of rejections of Algorithm \ref{alg: constant space} is $O(1)$ (or $O(\log m)$ with high probability).
Our main Theorem~\ref{th:main theorem} follows from the fact that the sequential shuffler of \cite{hiddenShuffle} uses constant space, its functions $\mathtt{pop}()$ and $\mathtt{empty}()$ run in constant time, and the while loop at Line \ref{line:while} of Algorithm \ref{alg: constant space} runs for at most $m$ iterations (less only in case of rejection) every time Algorithm \ref{alg: constant space} is executed.

\section{Counting Wheeler DFAs}\label{sec:DFAcount}

In this section, we use the WDFA characterization of Section~\ref{sec:algo} to give an exact formula for the number $|\mathcal D_{n, m, \sigma}|$ of WDFAs with $n$ nodes and $m$ edges on effective alphabet $[\sigma]$ with Wheeler order $1<2<\dots < n$. \ArxivOrCr{}{All proofs of this section can be found in the full version of the article~\cite{full-version}.} From our previous results, all we need to do is to compute the cardinalities of $\mathcal O_{n,m,\sigma}$ and $\mathcal I_O$. 

\begin{restatable}{lemma}{cardinality}
$|\mathcal O_{n,m,\sigma}| = \sum_{j=0}^\sigma (-1)^j \binom{\sigma}{j} \binom{n(\sigma-j)}{m}$.
\end{restatable}
\ArxivOrCr{
\begin{proof}
    Elements of $\mathcal O_{n,m,\sigma}$ are $n\times \sigma$ binary matrices with non-zero columns. For $j\in[\sigma]$, let us denote by $\alpha_j \subseteq \{0,1\}^{n\times \sigma}$ 
    the set of all $n\times \sigma$ binary matrices $O$
    such that $\|O_j\| = 0$.
    We have that
    \[
        |\mathcal O_{n,m,\sigma}|=\binom{n\sigma}{m} - \Big|\bigcup_{j=1}^\sigma \alpha_j\Big|.
    \]    
    Using the inclusion-exclusion principle, we obtain
    \begin{align*}
      \big|\bigcup_{j=1}^\sigma \alpha_j\big|   
      & = \sum_{j=1}^\sigma \left|\alpha_j\right| - \sum_{1\le i<j\le\sigma}\left|\alpha_i\cap \alpha_j\right| 
       + \sum_{1\le i<j<k\le\sigma} \left| \alpha_i\cap \alpha_j\cap \alpha_k \right|\\
       & \hspace{6cm} -\; \ldots \; + (-1)^{\sigma-1} \left| \alpha_1\cap \alpha_2\cap\dots \cap\alpha_\sigma \right|\\
      & = \sum_{j=1}^\sigma (-1)^{j-1} \binom{\sigma}{j} \binom{n(\sigma - j)}{m}.
    \end{align*}
    Therefore
    \[
    |\mathcal O_{n,m,\sigma}|=\binom{n\sigma}{m}-\sum_{j=1}^\sigma (-1)^{j-1} \binom{\sigma}{j} \binom{n(\sigma - j)}{m} = \sum_{j=0}^\sigma (-1)^j \binom{\sigma}{j} \binom{n(\sigma -j)}{m}
    \]
    and our claim follows. 
\end{proof}
}{This lemma is obtained via an inclusion-exclusion argument.}
From Algorithm \ref{alg: sample_I}, it is immediate that $|\mathcal I_O| = \binom{m-\sigma}{n-\sigma-1}$ for all $O \in \mathcal O_{n,m,\sigma}$ (see also the proof of Lemma~\ref{lem: uniform}). 
Since $r:\mathcal D_{n, m, \sigma} \rightarrow \mathcal R_{n, m, \sigma}$ is bijective (Corollary \ref{cor:r is bijective}), we obtain an exact formula for $|\mathcal D_{n, m, \sigma}|$: 

\begin{theorem}\label{th: number of WDFAs}
The number $|\mathcal D_{n, m, \sigma}|$ of WDFAs with set of nodes $[n]$ and $m$ transitions labeled from the effective alphabet $[\sigma]$, for which $1<2<\dots < n$ is a Wheeler order is
    $$|\mathcal D_{n, m, \sigma}| = \binom{m-\sigma}{n-\sigma-1}\sum_{j=0}^\sigma (-1)^j \binom{\sigma}{j} \binom{n(\sigma-j)}{m}.
    $$
\end{theorem}
\ArxivOrCr{
\begin{proof}
    Elements of $\mathcal O_{n,m,\sigma}$ are $n\times \sigma$ binary matrices with non-zero columns. For $j\in[\sigma]$, let us denote by $\alpha_j \subseteq \{0,1\}^{n\times \sigma}$ 
    the set of all $n\times \sigma$ binary matrices $O$
    such that $\|O_j\| = 0$.
    We have that
    \[
        |\mathcal O_{n,m,\sigma}|=\binom{n\sigma}{m} - \Big|\bigcup_{j=1}^\sigma \alpha_j\Big|.
    \]    
    Using the inclusion-exclusion principle, we obtain
    \begin{align*}
      \big|\bigcup_{j=1}^\sigma \alpha_j\big|   
      & = \sum_{j=1}^\sigma \left|\alpha_j\right| - \sum_{1\le i<j\le\sigma}\left|\alpha_i\cap \alpha_j\right| 
       + \sum_{1\le i<j<k\le\sigma} \left| \alpha_i\cap \alpha_j\cap \alpha_k \right|\\
       & \hspace{6cm} -\; \ldots \; + (-1)^{\sigma-1} \left| \alpha_1\cap \alpha_2\cap\dots \cap\alpha_\sigma \right|\\
      & = \sum_{j=1}^\sigma (-1)^{j-1} \binom{\sigma}{j} \binom{n(\sigma - j)}{m}.
    \end{align*}
    Therefore
    \[
    |\mathcal O_{n,m,\sigma}|=\binom{n\sigma}{m}-\sum_{j=1}^\sigma (-1)^{j-1} \binom{\sigma}{j} \binom{n(\sigma - j)}{m} = \sum_{j=0}^\sigma (-1)^j \binom{\sigma}{j} \binom{n(\sigma -j)}{m}
    \]
    and our claim follows. 
\end{proof}
}{}

\ArxivOrCr{We can re-formulate Theorem \ref{th: number of WDFAs} 
by lifting the requirement that the alphabet $\Sigma$ is effective. In order to do so, it is sufficient to fix as effective alphabet any nonempty subset of $\Sigma$. We immediately obtain: 
\begin{corollary}\label{cor: number of WDFAs non-effective}
    Let $\mathcal{\hat D}_{n, m, \sigma}$ denote the set of all WDFAs with set of nodes $[n]$ and $m$ transitions labeled from a fixed totally-ordered alphabet $\Sigma$ of cardinality $\sigma$, for which $1<2<\dots < n$ is a Wheeler order. Then: 
    $$
       |\mathcal{\hat D}_{n, m, \sigma}| =
       \sum_{k=1}^\sigma \binom{\sigma}{k} \binom{m-k}{n-k-1}\sum_{j=0}^k (-1)^j \binom{k}{j} \binom{n(k-j)}{m}. 
    $$
\end{corollary}
}{}
Using similar techniques, in the case where $\sigma$ is not arbitrarily close to $n$, i.e., $\sigma\le (1-\eps)\cdot n$ for some constant $\eps$, we moreover obtain a tight formula for the logarithm of the cardinality of $\mathcal{D}_{n,\sigma}=\bigcup_{m}\mathcal{D}_{n,m,\sigma}$, the set of all Wheeler DFAs with $n$ states over effective alphabet $[\sigma]$ and Wheeler order $1<2<\dots<n$:

\begin{restatable}{theorem}{WDFAbounds}\label{thm:bound n sigma}
    The following bounds hold:
    \begin{enumerate}
        \item $ \log |\mathcal D_{n, \sigma}| \geq n\sigma + (n - \sigma) \log \sigma - (n + \log \sigma)$, for any $n$ and $\sigma \leq n-1$, and 
        \item $\log |\mathcal D_{n, \sigma}| 
        \le n\sigma + (n - \sigma) \log \sigma + O(n)$, for any $n\ge 2/\eps$ and $\sigma \le (1 - \eps)\cdot n$, where $\eps$ is any desired constant such that $\eps\in (0, 1/2]$.
    \end{enumerate}
\end{restatable}

\ArxivOrCr{
\begin{proof}
    (1 --- Lower bound): Fix $m$ with $n-1\le m\le n\sigma$ and consider the set $\mathcal{D}'_{n,m,\sigma}\subseteq \mathcal{D}_{n,m,\sigma}$ of all WDFAs over the size-$\sigma$ alphabet with $n$ states and $m$ transitions and the source state having $\sigma$ outgoing transitions. Let us denote $\nu_m:= |\mathcal{D}'_{n,m,\sigma}|$.
    With our representation, these WDFAs have the only restriction that the first row of the out-matrix contains only 1s.
    Obviously, such an out-matrix has always at least one 1-bit in each column.
    Using our representation, we can thus obtain that the number of such WDFAs is
    \begin{align*}
        \nu_m
        &=\binom{n\sigma-\sigma}{m-\sigma}\binom{m-\sigma}{n-\sigma-1}
        =\frac{(n\sigma-\sigma)!}{(m-\sigma)!(n\sigma-m)!}\frac{(m-\sigma)!}{(n-\sigma-1)!(m-n+1)!}\\
        &= \frac{(n\sigma-\sigma)!}{(n\sigma-m)!(m-n+1)!(n-1-\sigma)!}.
    \end{align*}
    Now consider the following three-variate polynomial 
    \[
        f(x, y, z) = (x+y+z)^{n\sigma-\sigma} 
        = \sum_{(i, j, k)\in [n \sigma - \sigma]^3: i + j + k = n \sigma - \sigma}  \alpha_{i, j, k} \cdot x^i y^j z^k. 
    \]
    Clearly, $\alpha_{n\sigma-m, m-(n-1), (n-1)-\sigma}=\nu_m$ by the multinomial theorem. Now consider the univariate polynomial $f(1, 1, z) = (2+z)^{n\sigma-\sigma} = \sum_{k = 0}^{n\sigma - \sigma} \beta_k z^{k}$ and observe that 
    \begin{align*}
        \beta_{n - 1 - \sigma} 
        &= \sum_{(i, j)\in [n\sigma -\sigma]^2 : i + j = n\sigma - n + 1} \alpha_{i, j, n - 1 - \sigma}
        = \sum_{\ell = 0}^{n\sigma - n + 1} \alpha_{\ell, n\sigma - n + 1 - \ell, n - 1 - \sigma}\\
        &= \sum_{m = n - 1}^{n\sigma} \alpha_{n\sigma-m, m-(n-1), (n-1)-\sigma}
        = \sum_{m = n - 1}^{n\sigma} \nu_m
    \end{align*}
    On the other hand 
    \begin{align*}
        \beta_{n - 1 - \sigma} 
        = \binom{n\sigma-\sigma}{n-1-\sigma}2^{(n\sigma-\sigma)-(n-1-\sigma)}
        = \binom{n\sigma-\sigma}{n-1-\sigma}2^{n\sigma-n+1}.
    \end{align*}
    We thus obtain (by taking the base-2-logarithm)
    \begin{align*}
        \log |\mathcal{D}_{n,\sigma}|
        &\ge \log \sum_{m=n-1}^{n\sigma} \nu_m \\
        &\ge n\sigma - n + \log \binom{n\sigma-\sigma}{n-1-\sigma}\\
        &\ge n\sigma - n +\log \left(\Big( \frac{n\sigma-\sigma}{n-1-\sigma}\Big)^{n-1-\sigma}\right)\\
        &\ge n\sigma - n + (n - 1 - \sigma) \cdot \log\sigma \\
        &= n\sigma + (n - \sigma) \cdot \log \sigma - (n + \log \sigma)
    \end{align*}
     
    (2 --- Upper bound): Using our representation of WDFAs by bit-matrices and bit-vectors and the fact that $|\mathcal I_O|\le \binom{m - \sigma}{n - \sigma - 1}$, we get
    \begin{align*}
        |\mathcal{D}_{n,\sigma}|
        &= \sum_{m = n - 1}^{n\sigma}
        |\mathcal{D}_{n,m,\sigma}|
        \le \sum_{m = n - 1}^{n\sigma}
        |\mathcal O_{n, \sigma, m}| \cdot \binom{m - \sigma}{n - \sigma - 1}\\
        &\le \binom{n\sigma - \sigma}{n - \sigma - 1} \cdot \sum_{m = n - 1}^{n\sigma}
        |\mathcal O_{n, \sigma, m}|
        \le \binom{n\sigma - \sigma}{n - \sigma - 1} \cdot 2^{n\sigma}
        \le \Big(\frac{e (n\sigma - \sigma)}{n - \sigma - 1}\Big)^{n - \sigma} \cdot 2^{n\sigma}.
    \end{align*}
    Taking logarithms, we thus get 
    \begin{align*}
        \log |\mathcal{D}_{n,\sigma}|
        \le n \sigma + (n - \sigma)\cdot \log \Big(\frac{e (n\sigma - \sigma)}{n - \sigma - 1}\Big)
        \le n \sigma + (n - \sigma)\cdot \log \Big(\frac{e \sigma n}{n - \sigma - 1}\Big)
    \end{align*}
    Using the assumption that $\sigma \le (1 - \eps) n$ for some constant $\eps\le 1/2$, we obtain that the logarithm above can be upper bounded by $\log (e \sigma\cdot \frac{n}{\eps n - 1}\big)$. Now assuming that $n\ge 1/\eps$, it holds that the function $f(n)=\frac{n}{\eps n - 1}$ is monotonically decreasing and hence $f(n)\le f(\frac{2}{\eps})=\frac{2}{\eps}=O(1)$ for all $n\ge \frac{2}{\eps}$. Altogether, 
    \[
        \log |\mathcal{D}_{n,\sigma}| 
        \le n \sigma + (n - \sigma) \cdot (\log \sigma + O(1)) 
        = n \sigma + (n - \sigma) \cdot \log \sigma + O(n).\qedhere
    \]
\end{proof}
}{\medskip}

Note that $\log |\mathcal{D}_{n,\sigma}|$ is the information-theoretic worst-case number of bits necessary (and sufficient) to encode a WDFA from $\mathcal{D}_{n,\sigma}$. Our Theorem \ref{thm:bound n sigma} states that, up to an additive $\Theta(n)$ number of bits, this value is of $ n\sigma + (n - \sigma) \log \sigma$ bits. As a matter of fact, our encoding $r(D)=(O,I)$ of Section \ref{sec:uniform WDFAs}, opportunely represented using succinct bitvectors \cite{RRR}, achieves this bound up to additive lower-order terms and supports efficient navigation of the transition relation. 

\section{Implementation}

We implemented our uniform WDFA sampler in \texttt{C++}.\footnote{Implementation available at \url{https://github.com/regindex/Wheeler-DFA-generation}.} 
We tested our implementation by generating WDFAs with a broad range of parameters:
$n \in \{10^6 \cdot 2^{i}: i=0,\ldots ,6 \}$, $m \in \{n \cdot 2^{i} - 1: i=0,\ldots ,7 \}$  
and $\sigma = 128$. 
To analyze the impact of streaming to disk on the running time, we tested two versions of our code: (1) We stream the resulting WDFA to disk (SSD). (2) We stream the WDFA to a pre-allocated vector residing in internal memory. Note that constant working space is achieved only in case (1). 
Our experiments were run on a server with Intel(R) Xeon(R) W-2245 CPU @ 3.90GHz with 8 cores, 128 gigabytes of RAM, 512 gigabytes of SSD, running Ubuntu 18.04 LTS 64-bit. Working space was measured with \texttt{/usr/bin/time} (Resident set size).

\begin{figure*}[ht!]
    \centering
 	\includegraphics[width=0.44\textwidth, trim={5.5mm 5.5mm 5.0mm 5.5mm}, clip]{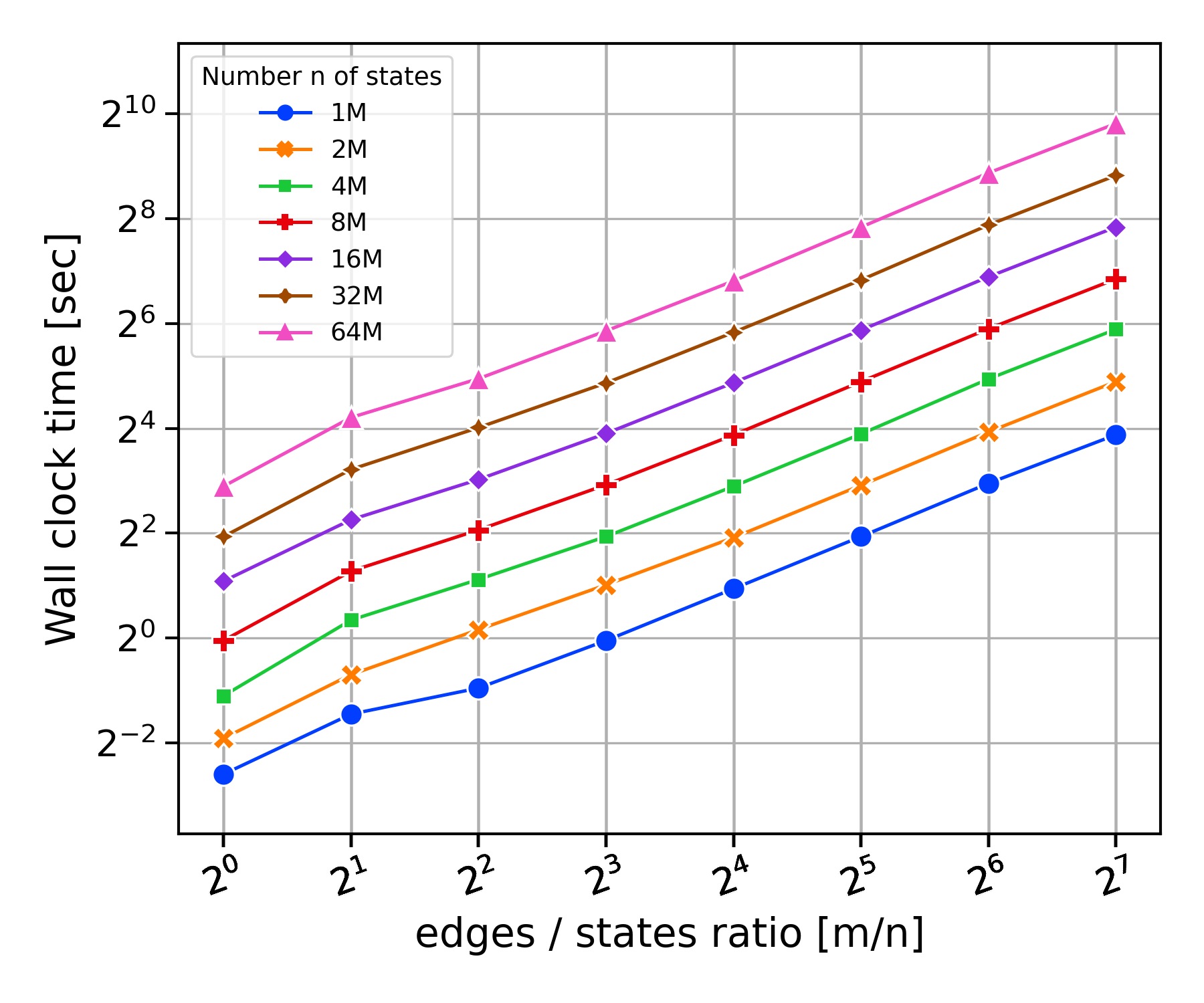}
    \hspace{3mm}
    \includegraphics[width=0.44\textwidth, trim={5.5mm 5.5mm 5.0mm 5.5mm}, clip]{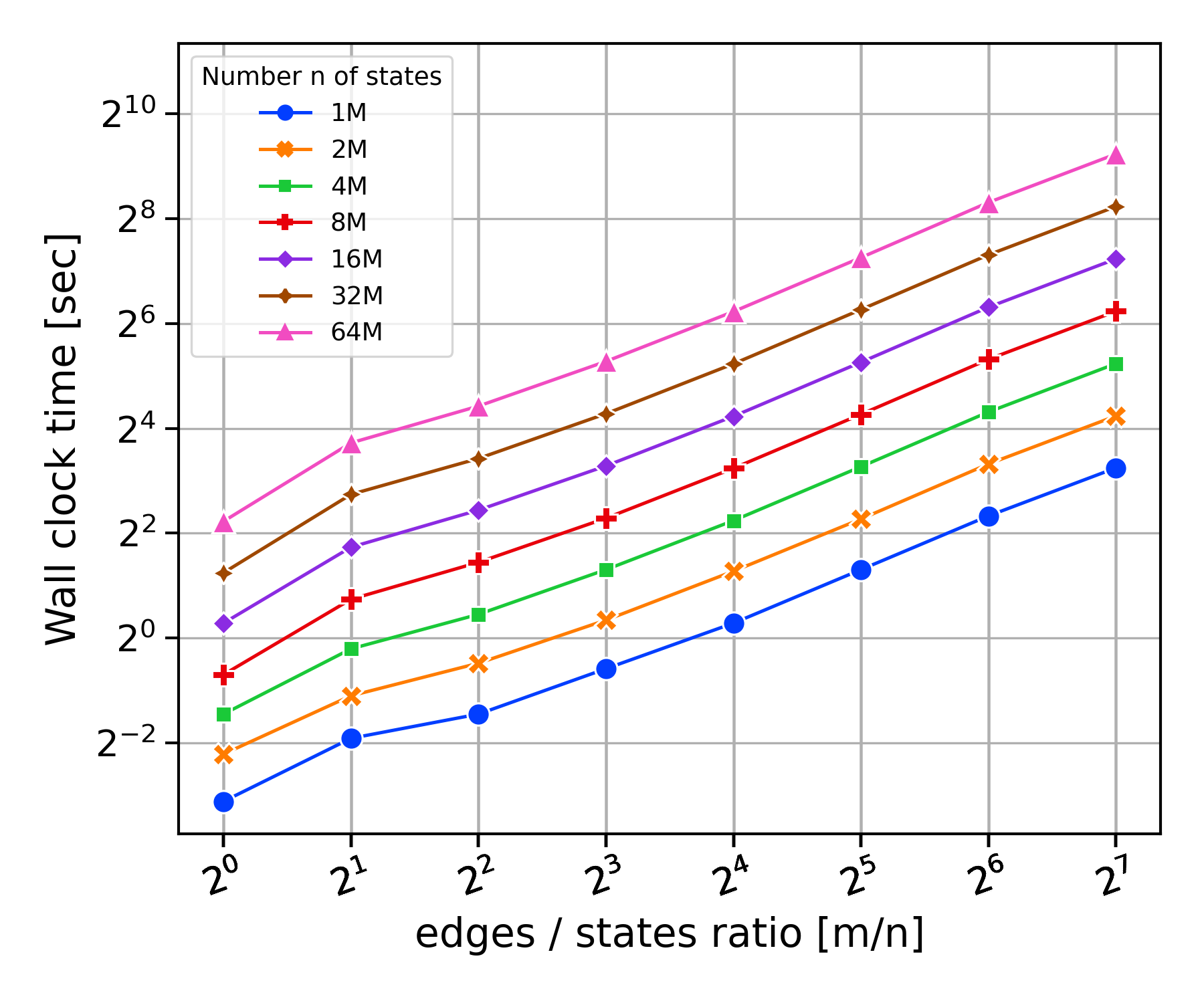}
    \caption{
    Wall clock time for generating random WDFAs using Algorithm~\ref{alg: constant space}. Left: running time for the algorithm in case (1), i.e., streaming the resulting WDFAs to disk. Right: running time in case (2), i.e., storing WDFAs in internal memory. 
    }
    \label{fig:exp1}
\end{figure*}

Figure \ref{fig:exp1} shows the running time of both variants (left: (1) streaming to SSD; right: (2) streaming to RAM). Both versions exhibit a linear running time behavior, albeit with a different multiplicative constant. The algorithm storing the WDFA in internal memory is between 1.2 and 1.7 times faster than the version streaming the WDFA to the disk (the relatively small difference is due to the fact that we used an SSD). We measured a throughput of at least 5.466.897 and 7.525.794 edges per second for the two variants, respectively. 
In our experiments we never observed a rejection: this is due to the fact that $\sigma \ll m$, making it extremely likely to generate bit-matrices $O$ containing at least one set bit in each column. 

As far as space usage is concerned, version (1), i.e., streaming the WDFA to disk, always used about 4 MB of internal memory, independently from the input size (this memory is always required to load the \texttt{C++} libraries). 
This confirms the constant space usage of our algorithm, also experimentally.
As expected, the space usage of version (2) is linear with the input's size. Nevertheless, both algorithms are extremely fast in practice: in these experiments, the largest automaton consisting of 64 million states and more than 8 billion edges was generated in about 15 and 10 minutes with the first and second variant, respectively.

\newpage
\bibliography{references}

\begin{thebibliography}{10}

\bibitem{AlankoRegular2020}
Jarno Alanko, Giovanna D'Agostino, Alberto Policriti, and Nicola Prezza.
\newblock {Regular Languages Meet Prefix Sorting}.
\newblock In {\em Proceedings of the Thirty-First Annual ACM-SIAM Symposium on Discrete Algorithms}, SODA '20, page 911–930, USA, 2020. Society for Industrial and Applied Mathematics.

\bibitem{ALANKO2021104820}
Jarno Alanko, Giovanna D'Agostino, Alberto Policriti, and Nicola Prezza.
\newblock Wheeler languages.
\newblock {\em Information and Computation}, 281:104820, 2021.
\newblock URL: \url{https://www.sciencedirect.com/science/article/pii/S0890540121001504}.

\bibitem{alanko2019tunneling}
Jarno Alanko, Travis Gagie, Gonzalo Navarro, and Louisa~Seelbach Benkner.
\newblock {Tunneling on wheeler graphs}.
\newblock In {\em 2019 Data Compression Conference (DCC)}, pages 122--131. IEEE, 2019.

\bibitem{burrows1994block}
Michael Burrows and David~J Wheeler.
\newblock A block-sorting lossless data compression algorithm.
\newblock Technical Report 124, Digital Equipment Corporation, 1994.

\bibitem{chao2022wgt}
Kuan-Hao Chao, Pei-Wei Chen, Sanjit~A Seshia, and Ben Langmead.
\newblock {WGT: Tools and algorithms for recognizing, visualizing and generating Wheeler graphs}.
\newblock {\em bioRxiv}, pages 2022--10, 2022.

\bibitem{ConteMatchingStatistics2023}
Alessio Conte, Nicola Cotumaccio, Travis Gagie, Giovanni Manzini, Nicola Prezza, and Marinella Sciortino.
\newblock Computing matching statistics on wheeler dfas.
\newblock In {\em 2023 Data Compression Conference (DCC)}, pages 150--159, 2023.
\newblock \href {https://doi.org/10.1109/DCC55655.2023.00023} {\path{doi:10.1109/DCC55655.2023.00023}}.

\bibitem{DAgostinoOrdering2023}
Giovanna D'Agostino, Davide Martincigh, and Alberto Policriti.
\newblock {Ordering regular languages and automata: Complexity}.
\newblock {\em Theoretical Computer Science}, 949:113709, 2023.
\newblock URL: \url{https://www.sciencedirect.com/science/article/pii/S0304397523000221}.

\bibitem{egidi2022space}
Lavinia Egidi, Felipe~A Louza, and Giovanni Manzini.
\newblock {Space efficient merging of de Bruijn graphs and Wheeler graphs}.
\newblock {\em Algorithmica}, 84(3):639--669, 2022.

\bibitem{gagie2022representing}
Travis Gagie.
\newblock {On Representing the Degree Sequences of Sublogarithmic-Degree Wheeler Graphs}.
\newblock In {\em String Processing and Information Retrieval: 29th International Symposium, SPIRE 2022, Concepci{\'o}n, Chile, November 8--10, 2022, Proceedings}, pages 250--256. Springer, 2022.

\bibitem{gagie:tcs17:wheeler}
Travis Gagie, Giovanni Manzini, and Jouni Sir\'en.
\newblock {Wheeler graphs: A framework for BWT-based data structures}.
\newblock {\em Theoretical Computer Science}, 698:67--78, 2017.
\newblock \href {https://doi.org/10.1016/j.tcs.2017.06.016} {\path{doi:10.1016/j.tcs.2017.06.016}}.

\bibitem{gibney2022complexity}
Daniel Gibney and Sharma~V Thankachan.
\newblock {On the complexity of recognizing wheeler graphs}.
\newblock {\em Algorithmica}, 84(3):784--814, 2022.

\bibitem{goga2022prefix}
Adri{\'a}n Goga and Andrej Bal{\'a}{\v{z}}.
\newblock {Prefix-Free Parsing for Building Large Tunnelled Wheeler Graphs}.
\newblock In {\em 22nd International Workshop on Algorithms in Bioinformatics}, 2022.

\bibitem{Knuth98}
Donald~Ervin Knuth.
\newblock {\em The art of computer programming, Volume {II:} Seminumerical Algorithms, 3rd Edition}.
\newblock Addison-Wesley, 1998.
\newblock URL: \url{https://www.worldcat.org/oclc/312898417}.

\bibitem{nicaud14:randomDFA}
Cyril Nicaud.
\newblock {Random Deterministic Automata}.
\newblock In {\em Proceedings of the 39th International Symposium on Mathematical Foundation of Computer Science (MFCS)}, pages 5--23, 2014.

\bibitem{RRR}
Rajeev Raman, Venkatesh Raman, and S.~Srinivasa Rao.
\newblock Succinct indexable dictionaries with applications to encoding k-ary trees and multisets.
\newblock In {\em Proceedings of the Thirteenth Annual ACM-SIAM Symposium on Discrete Algorithms}, SODA '02, page 233–242, USA, 2002. Society for Industrial and Applied Mathematics.

\bibitem{hiddenShuffle}
Michael Shekelyan and Graham Cormode.
\newblock Sequential random sampling revisited: Hidden shuffle method.
\newblock In Arindam Banerjee and Kenji Fukumizu, editors, {\em Proceedings of The 24th International Conference on Artificial Intelligence and Statistics}, volume 130 of {\em Proceedings of Machine Learning Research}, pages 3628--3636. PMLR, 13--15 Apr 2021.
\newblock URL: \url{https://proceedings.mlr.press/v130/shekelyan21a.html}.

\end{thebibliography}









\end{document}